\documentclass{amsart}
\usepackage{amssymb,amscd,amsthm}
\usepackage{latexsym}
\date{\today}

\newcommand{\Z}{{\mathbb Z}}
\newcommand{\R}{{\mathbb R}}
\newcommand{\C}{{\mathbb C}}

\newtheorem{theorem}{Theorem}

\newtheorem{lemma}{Lemma}
\newtheorem{prop}{Proposition}
\newtheorem{coro}{Corollary}

\def\e{\varepsilon}

\begin{document}
\title[Fractal Dimension of the Spectrum of the Fibonacci Hamiltonian]{The
Fractal
Dimension of the Spectrum of the Fibonacci Hamiltonian}

\author[D.\ Damanik]{David Damanik}

\address{Department of Mathematics, Rice University, Houston, TX~77005, USA}

\email{damanik@rice.edu}

\author[M.\ Embree]{Mark Embree}

\address{Computational and Applied Mathematics, Rice University, Houston,
TX~77005,
USA}

\email{embree@rice.edu}

\author[A.\ Gorodetski]{Anton Gorodetski}

\address{Mathematics 253-37, California Institute of Technology, Pasadena, CA
91125,
USA}

\email{asgor@caltech.edu}

\author[S.\ Tcheremchantsev]{Serguei Tcheremchantsev}

\address{Universit\'{e} d'Orl\'{e}ans, Laboratoire MAPMO,
CNRS-UMR 6628, B.P.~6759, F-45067
Orl\'{e}ans Cedex,
France}

\email{serguei.tcheremchantsev@univ-orleans.fr}

\thanks{D.\ D.\ was supported in part by NSF grant DMS--0653720.
M.\ E.\ was supported by NSF grant DMS--CAREER--0449973.}

\begin{abstract}
We study the spectrum of the Fibonacci Hamiltonian and prove upper
and lower bounds for its fractal dimension in the large coupling
regime. These bounds show that as $\lambda \to \infty$, $\dim
(\sigma(H_\lambda)) \cdot \log \lambda$ converges to an explicit
constant ($\approx 0.88137$). We also discuss consequences of these
results for the rate of propagation of a wavepacket that evolves
according to Schr\"odinger dynamics generated by the Fibonacci
Hamiltonian.
\end{abstract}

\maketitle

\section{Introduction}

The Fibonacci Hamiltonian is a discrete one-dimensional
Schr\"odinger operator
$$
[H u](n) = u(n+1) + u(n-1) + V(n) u(n)
$$
in $\ell^2(\Z)$. The potential $V : \Z \to \R$ is given by
\begin{equation}\label{fibpot}
V(n) = \lambda \chi_{[1-\phi^{-1},1)}(n \phi^{-1} \! +
\theta\!\!\! \mod 1),
\end{equation}
where $\lambda > 0$ is the coupling constant, $\phi$ is the golden
mean,
\begin{equation}\label{goldenratio}
\phi = \frac{\sqrt{5}+1}{2} = 1+ \cfrac{1}{1 + \cfrac{1}{1 +
\cfrac{1}{1+\cdots}}},
\end{equation}
and $\theta \in [0,1)$ is the phase.

This operator is important for both physical and mathematical
reasons. On the one hand, it is the most popular quantum model of
a one-dimensional quasicrystal, that is, a structure that shares
many features with one displaying global order, but which in fact
lacks global translation invariance. On the other hand, this
operator has zero-measure Cantor spectrum and all spectral
measures are purely singular continuous. These properties had been
regarded as ``exotic'' in the context of general Schr\"odinger
operators up until the 1980s, but for this operator family, they
occur persistently for all parameter values. The Fibonacci
Hamiltonian has been heavily studied since the early 1980s; see
\cite{D} for a recent review of the results obtained for it and
related models.

Let us recall some specific results and references that will be
important for what follows. It is known that the spectrum of the
Fibonacci Hamiltonian is independent of $\theta$ (see, e.g.,
\cite{bist}). This follows quickly from strong convergence once one
realizes that for each pair $\theta,\tilde \theta$, there is a
sequence $n_k \to \infty$ such that $\theta + n_k \phi^{-1}$
converges to $\tilde \theta$ in $\R / \Z$ ``from the right.'' We
denote this common spectrum by $\Sigma_\lambda$,
$$
\Sigma_\lambda = \sigma(H) \quad \text{ for every } \theta \in
[0,1).
$$

It is natural to study the spectrum as a set. As we shall see, such a
study is also motivated by the consequences one can draw for the
long-time behavior of the solution of the time-dependent
Schr\"odinger equation.

It has been shown by S\"ut\H{o} that the spectrum always has zero
Lebesgue measure \cite{su2},
\begin{equation}\label{zmspec}
\mathrm{Leb} ( \Sigma_\lambda ) = 0 \quad \text{ for every }
\lambda > 0.
\end{equation}
This immediately implies the absence of absolutely continuous
spectrum for all parameter values.\footnote{Historically, these
two properties were established in the reverse order. The methods
used by Kotani \cite{ko} in the proof of absence of absolutely
continuous spectrum (for almost every $\theta$) were the key to
proving zero measure spectrum.} It was later seen that one also
has the absence of point spectrum for all parameter values; see
S\"ut\H{o} \cite{su}, Hof-Knill-Simon \cite{hks}, and Kaminaga
\cite{ka} for partial results and Damanik-Lenz \cite{dl} for the
full result. Thus, the Fibonacci model exhibits purely singular
continuous spectrum that is very rigid in the sense that it is not
affected by a change of the defining parameters.

The result on zero measure spectrum, \eqref{zmspec}, naturally
leads one to ask about the dimension of this set. There are
several popular ways to measure the fractal dimension of a nowhere
dense subset of the real line.

Let us recall the definition of two of these dimensions. Suppose
we are given a bounded and infinite set $S \subseteq \R$. A
$\delta$-cover of $S$ is a countable union of real intervals, $\{
I_m \}_{m \ge 1}$, such that each of these intervals has length
bounded by $\delta > 0$. For $\alpha \in [0,1]$, let
$$
h^\alpha(S) = \lim_{\delta \to 0} \; \inf_{\delta\text{-covers}}
\; \sum_{m \ge 1} |I_m|^\alpha.
$$
It is clear that the limit exists in $[0,\infty]$. Moreover, if
$h^\alpha (S) = 0$ for some $\alpha$, then $h^{\alpha '}(S) = 0$
for every $\alpha ' > \alpha$. Similarly, if $h^\alpha (S) =
\infty$ for some $\alpha$, then $h^{\alpha '}(S) = \infty$ for
every $\alpha ' < \alpha$. Thus, the following quantity is
well-defined:
$$
\dim_H (S) = \inf \{\alpha : h^\alpha(S) < \infty \} = \sup
\{\alpha : h^\alpha(S) = \infty \}.
$$
The number $\dim_H (S) \in [0,1]$ is called the Hausdorff
dimension of the set $S$.

A different way to measure the fractal
dimension of $S$ is via the box counting dimension.
The lower box counting dimension is defined as follows:
$$
\dim_B^-(S) = \liminf_{\varepsilon \to 0} \frac{\log
N_S(\varepsilon)}{\log 1/\varepsilon},
$$
where
$$
N_S(\varepsilon) = \# \{ j \in \Z : [j\varepsilon,
(j+1)\varepsilon) \cap S \not= \emptyset \}.
$$
The upper box counting dimension, $\dim_B^+(S)$, is defined
similarly, with the $\liminf$ replaced by a $\limsup$. When
$\dim_B^+(S)$ and  $\dim_B^-(S)$ are equal, we denote their common
value by $\dim_B(S)$ and call this number the box counting
dimension of $S$. These dimensions are related by the inequalities
$$
\dim_H(S) \le \dim_B^-(S) \le \dim_B^+(S).
$$
In general, both inequalities may be strict; see, for example,
\cite[pp.~76--77]{matt}.

The main goal of this paper is to study the fractal dimension of
the spectrum of the Fibonacci Hamiltonian. The following result
shows that for sufficiently large coupling, the dimensions just
introduced coincide.

\begin{theorem}\label{dimeqthm}
Suppose that $\lambda \ge 16$.
Then the box counting dimension of $\Sigma_\lambda$ exists and obeys
$$
\dim_B (\Sigma_\lambda) = \dim_H (\Sigma_\lambda).
$$
\end{theorem}

While this theorem has not appeared in print explicitly before, it
does follow quickly from a combination of known results. We
present the relevant facts in the appendix.

Theorem~\ref{dimeqthm} is useful because it will allow us to
obtain precise asymptotics for the fractal dimension of
$\Sigma_\lambda$ as $\lambda \to \infty$. The reason for this is
the following. The box counting dimension is easier to bound from
below, while the Hausdorff dimension is easier to bound from
above. Consequently, we will prove a lower bound for the box
counting dimension in Section~\ref{boxsec} and an upper bound for
the Hausdorff dimension in Section~\ref{haussec}.

It is known how to describe the spectrum of the Fibonacci
Hamiltonian in terms of the spectra of canonical periodic
approximants. We will recall this in Section~\ref{bandsec}. The
general theory of periodic discrete one-dimensional Schr\"odinger
operators shows that the spectrum of a such a periodic operator is
always given by a finite union of compact intervals. Our crucial
new insight is a way to describe the asymptotic distribution of
bandwidths in these periodic spectra. In this description, the
following function plays an important role. Define
\begin{eqnarray*}
   f(x) = \frac{1}{x} \left[\right. \kern-20pt&& (2-3x) \log 2 + (1-x) \log (1-x) \\[-.25em]
          && \kern10pt {}-(2x-1)\log (2x-1) - (2-3x) \log (2-3x) \left.\kern-2pt\right]
\end{eqnarray*}
on the interval $(\frac{1}{2},\frac{2}{3})$. Setting
$f(\frac{1}{2}) = \log 2$ and $f(\frac{2}{3}) = 0$, it is not hard
to see that $f$ extends to a continuous function on
$[\frac{1}{2},\frac{2}{3}]$, and with the aid of symbolic computation
one can confirm that it takes its maximum at the unique point
$$
x^* = {12-2\sqrt{2} \over 17} = 0.5395042867796\ldots,
$$
with
$$
f^* = f(x^*) = \log(1+\sqrt{2}) = 0.8813735870195\ldots.
$$
Write
$$
S_u(\lambda) = 2\lambda + 22
$$
and
$$
S_l(\lambda) = \frac{1}{2}
          \left((\lambda - 4) + \sqrt{(\lambda - 4)^2 - 12} \right).
$$
With these functions of (sufficiently large) $\lambda$ we can now
state the bounds on the fractal dimension of $\Sigma_\lambda$ that
we will prove in Sections~\ref{boxsec} and \ref{haussec},
respectively.

\begin{theorem}\label{mainthm}
{\rm (a)} Suppose $\lambda > 4$. Then
$$
\dim_B^-(\Sigma_\lambda) \ge \frac{f^*}{\log S_u(\lambda)}.
$$
{\rm (b)} Suppose $\lambda \ge 8$. Then
$$
\dim_H (\Sigma_\lambda) \le \frac{f^*}{\log S_l(\lambda)}.
$$
\end{theorem}

Since both $S_u(\lambda)$ and $S_l(\lambda)$ behave asymptotically
like $\log \lambda$, we obtain the following result as an immediate
consequence. We write $\dim$ for either $\dim_H$ or $\dim_B$, which
is justified by Theorem~\ref{dimeqthm}.

\begin{coro}
We have
$$
\lim_{ \lambda \to \infty } \dim (\Sigma_\lambda) \cdot \log
\lambda = f^*.
$$
\end{coro}

In particular, we see that the constant $f^*$ is the best possible
in both bounds in Theorem~\ref{mainthm}. Let us compare our results
with previously known ones. To facilitate this, we introduce
$$
f^\# = \frac{f^*}{\log \phi}
     = \frac{\log(1+\sqrt{2})}{\log \phi} = 1.8315709239073\ldots,
$$
so that our results can be summarized as follows. We have for
$\lambda\ge 16$,
\begin{equation}\label{ourbounds}
f^\# \frac{\log \phi}{\log S_u(\lambda)} \le \dim_B (\Sigma_\lambda)
= \dim_H (\Sigma_\lambda) \le f^\# \frac{\log \phi}{\log
S_l(\lambda) },
\end{equation}
and therefore the asymptotic behavior is
$$
\dim(\Sigma_\lambda) \sim f^\# \, \frac{\log \phi}{\log \lambda}.
$$

As we will see below, there are two competing scaling processes,
one scaling with $\phi$ (the Fibonacci numbers) and one scaling
with $\lambda$ (the inverse of the width of a band in the
approximating periodic spectra). Thus, it is natural to write a
bound in the form ``constant times $\frac{\log \phi}{\log
\lambda}$'' and then to optimize the constant.

Raymond \cite{r} proved an upper bound for $\dim_H
(\Sigma_\lambda)$ that has a $2$ in place of our $f^\#$ in
\eqref{ourbounds}. A simplified version of our approach (which we
will comment on later in the paper) quickly gives a lower bound
with $f^\#$ replaced in \eqref{ourbounds} by $1.5$ and an upper
bound with $f^\#$ replaced in \eqref{ourbounds} by $2$; the latter
being Raymond's result. These numbers appear naturally in this
context and are associated with the support of a certain discrete
probability distribution. A more detailed study of this
distribution then led us to the discovery of $f^\#$, which
describes the actual asymptotic behavior of the fractal dimension
of the spectrum as we saw above.

Lower bounds for the dimension of the spectrum were initially
obtained as a consequence of certain continuity properties of the
spectral measures with respect to certain Hausdorff measures.
Since the spectral measures are supported on the spectrum, one can
obtain a lower bound for the Hausdorff dimension of the spectrum
in this way. As mentioned above, this also bounds the box counting
dimension from below by general principles. Spectral Hausdorff
continuity results for the Fibonacci Hamiltonian were shown in
\cite{d1,dkl,jl,kkl}. The best lower bound that has been obtained
in this way can be found in \cite{kkl} and it reads
$$
\dim_H(\Sigma_\lambda) \ge \frac{2\kappa}{\kappa +
\zeta(\lambda)},
$$
where\footnote{Notice that there is a typo in \cite{kkl}. They
have $\kappa = \log \frac{\sqrt{17}}{20 \log \phi}$,
a negative number!}
$$
\kappa = \frac{\log \left( \frac{\sqrt{17}}{4} \right)}{5 \log \phi}
\approx 0.0126
$$
and
$$
\zeta(\lambda) = \frac{6 \log \sqrt{5}}{\log \phi} \left( \log
\lambda + O(1) \right).
$$
Thus, for $\lambda$ large, this gives a lower bound for
$\dim_H(\Sigma_\lambda)$ as in \eqref{ourbounds}, but with $f^\#$
replaced by
\begin{equation}\label{kklbound}
\frac{2\frac{\log \left( \frac{\sqrt{17}}{4} \right)}{5 \log
\phi}}{6\log \sqrt{5}} \approx 0.00188.
\end{equation}

Liu and Wen \cite{lw} then extended the approach employed by
Raymond. They study the case of general frequencies.
Specialized to the Fibonacci case, their result shows that
for $\lambda > 20$,
$$
\frac{\log 2}{10 \log 2 + 3 \log \left( 4 (\lambda - 8) \right)}
\le \dim_H (\Sigma_\lambda) \le \frac{\log 3}{\log \left(
\frac{\lambda - 8}{3} \right) }.
$$
Let us discuss this result in the large coupling limit. Since
$$
\frac{\log 3}{\log \phi} \approx 2.28301 > 2,
$$
the upper bound does not improve Raymond's result. The lower bound
has a constant coefficient
$$
\frac{\log 2}{3 \log \phi} \approx 0.48013
$$
in front of $\frac{\log \phi}{\log \lambda}$, a significant
improvement over the result that can be extracted from \cite{kkl}.
Again, by our result, the optimal constant is $f^\# \approx
1.83157$.

\bigskip

Our interest in obtaining the optimal constant $f^\#$ does not
only stem from natural curiosity. An interesting and
mathematically challenging problem is to study the spreading of a
wavepacket in a quantum system in the case where the initial state
has a purely singular continuous spectral measure. This is the
case for every initial state from $\ell^2(\Z)$ for a system
governed by the Fibonacci Hamiltonian. One is often especially
interested in the spreading of a wavepacket that is initially
localized on just one site.

That is, with $H$ as above, we consider $\psi (t) = e^{-itH}
\delta_1$ and study its spreading via the time-averaged outside
probabilities
$$
P_r (N,T)= \sum_{n > N} \frac{2}{T} \int_0^\infty
e^{-\frac{2t}{T}} \left| \left\langle e^{-itH} \delta_1, \delta_n
\right\rangle \right|^2 \, dt
$$
and
$$
P_l (N,T)= \sum_{n < -N} \frac{2}{T} \int_0^\infty
e^{-\frac{2t}{T}} \left| \left\langle e^{-itH} \delta_1, \delta_n
\right\rangle \right|^2 \, dt.
$$
Let $P(N,T) = P_l (N,T) + P_r (N,T)$ and define
$$
S^-(\alpha) = - \liminf_{T \to \infty} \frac{\log P(T^\alpha -2,
T) }{\log T}
$$
and
$$
S^+(\alpha) = - \limsup_{T \to \infty} \frac{\log P(T^\alpha -2,
T) }{\log T}.
$$
For every $\alpha$, $0 \le S^+ (\alpha) \le S^- (\alpha) \le
\infty$. These numbers control the power decaying tails of the
wavepacket. In particular, the following critical exponents are of
interest:
$$
\alpha_u^\pm  = \sup \{ \alpha \ge 0  :  S^\pm (\alpha) < \infty
\}.
$$
One can interpret $\alpha_u^\pm$ as the rates of propagation of
the fastest (polynomially small) part of the wavepacket; compare
\cite{GKT}. In particular, if $\alpha > \alpha_u^+$, then
$P(T^\alpha, T)$ goes to $0$ faster than any inverse power of $T$,
and if $\alpha > \alpha_u^-$, then there is a sequence of times
$T_k \to \infty$ such that $P(T_k^\alpha, T_k)$ goes to $0$ faster
than any inverse power of $T_k$.

In Section~\ref{qdsec} we prove a result for general
Schr\"odinger operators on $\ell^2(\Z)$ that will imply the
following consequence for the Fibonacci Hamiltonian.

\begin{theorem}\label{thm3}
For every $\lambda > 0$ and every $\theta \in [0,1)$, we have that
$$
\alpha_u^\pm \ge \dim_B^\pm (\Sigma_\lambda).
$$
Consequently, for $\lambda > 4$ and every $\theta$, we have
$$
\alpha_u^\pm \ge \frac{f^*}{\log S_u(\lambda)}.
$$
\end{theorem}

To discuss this result in the large coupling limit, let us be
slightly inaccurate\footnote{The precise statement is that
$\liminf_{\lambda \to \infty} \alpha_u^\pm  \, \frac{\log
\lambda}{\log \phi} \ge f^\#$.}  and write
\begin{equation}\label{ourqdbound}
\alpha_u^\pm \ge f^\# \, \frac{\log \phi}{\log \lambda}.
\end{equation}
There are two main previous approaches to quantum dynamical lower
bounds for the Fibonacci model. The first is based on spectral
continuity and the papers \cite{d1,dkl,jl,kkl} contain results
obtained in this way. For $\theta = 0$, the best bound is contained
in \cite{kkl} and it has \eqref{ourqdbound} with $f^\# \approx
1.83157$ replaced by \eqref{kklbound}, that is, $\approx\kern-2pt 0.00188$. For
other values of $\theta$, the best bound can be found in \cite{dkl}
and the constant in this bound is even smaller.

The other approach is based on complex energy methods and the
Plancherel Theorem; see \cite{DT} (and also \cite{DT1} for a way
to combine the two approaches). For $\theta = 0$, the paper
\cite{DT} has \eqref{ourqdbound} with $f^\#$ replaced by
$1/6$. It is possible to treat general $\theta$ along the
same lines using \cite{dl2}, but the dynamical lower bound has a
somewhat smaller constant in the general case.

Thus, on the one hand, our result improves the constant from the
previously best value $1/6$ to $f^\# \approx 1.83157$ and,
on the other hand, this is the best one can do using the
method put forth in this paper.
We would like to mention that \cite{DT2} contains the following upper bound for
$\alpha_u^+$, which holds for
$\lambda \ge 8$,
$$
\alpha_u^+ \le 2  \frac{\log \phi}{\log S_l(\lambda)}.
$$
While we know that the dimension of the spectrum indeed behave like
$f^\# \,  \frac{\log \phi}{\log \lambda}$ in the large coupling
limit, we expect that the dynamical quantities $\alpha_u^\pm$ behaves
like $2 \, \frac{\log \phi}{\log \lambda}$ in the large coupling
limit. That is, we expect the following to hold,
$$
\lim_{\lambda \to \infty} \alpha_u^\pm \, \frac{\log \lambda}{\log
\phi} = 2.
$$
The reason for this is that the spreading of the fastest part of
the wavepacket is determined by the ``most continuous'' part of
the spectral measure and, in this case, by the region in the
spectrum that is ``the thickest.'' Since we will see that there is
indeed a small region that is thickest in a natural sense, the
factor $2$ will then appear naturally. A forthcoming publication,
\cite{DT3}, will deal with this issue using ideas and results from
\cite{DT2}.

\section{The Band-Gap Structure of the Approximating Periodic
Spectra}\label{bandsec}

In this section we describe the canonical coverings of
$\Sigma_\lambda$ by (unions of) periodic spectra and the
hierarchical structure of these sets. We use the combinatorics of
this description to derive detailed results about the distribution
of bandwidths in these spectra.

For $E \in \R$ and $\lambda > 0$, we define a sequence of numbers
$x_k = x_k(E,\lambda)$ as follows.
\begin{equation}\label{tracemap}
x_{-1} = 2 , \quad x_0 = E , \quad x_1 = E - \lambda , \quad
x_{k+1} = x_k x_{k-1} - x_{k-2} \quad \text{ for } k \ge 1.
\end{equation}
Using this recurrence and the initial values, one can quickly check
(see S\"ut\H{o} \cite{su}) that
\begin{equation}\label{invariant}
x_{k+1}^2 + x_k^2 + x_{k-1}^2 - x_{k+1} x_k x_{k-1} = 4 +
\lambda^2 \; \mbox{ for every } k \ge 0.
\end{equation}
For fixed $\lambda > 0$, define
$$
\sigma_k = \{ E \in \R : |x_k (E, \lambda)| \le 2\}.
$$
The set $\sigma_k$ is actually equal to the spectrum of the
Schr\"odinger operator $H$ whose potential $V_k$ results from $V$
(with $\theta = 0$) by replacing $\phi^{-1}$ with $F_{k-1}/F_k$
(cf.~\cite{su}). Here, $\{F_k\}_{k \ge 0}$ denotes the sequence of
Fibonacci numbers, that is,
$$
F_0 = 1, \; F_1 = 1, \; F_{k+1} = F_k + F_{k-1} \mbox{ for } k \ge
1.
$$
Hence, $V_k$ is periodic and $\sigma_k$ consists of $F_k$ bands
(closed intervals). S\"ut\H{o} also proved that
\begin{equation}\label{suto1} 
\left(\sigma_{k-1} \cup \sigma_k\right) \supset \left( \sigma_k \cup
\sigma_{k+1}\right)
\end{equation}
and
\begin{equation}\label{suto2} 
\Sigma_\lambda = \bigcap_{k \ge 1} (\sigma_k \cup \sigma_{k+1}) = \{
E : \{ x_k \} \text{ is a bounded sequence} \}.
\end{equation}

From now on, we assume
\begin{equation}\label{four}
\lambda > 4,
\end{equation}
since we will make critical use of the fact that in this case, it
follows from the invariant \eqref{invariant} that three
consecutive $x_k$'s cannot all be bounded in absolute value
by~$2$:
\begin{equation}\label{critical}
\sigma_k \cap \sigma_{k+1} \cap \sigma_{k+2} = \emptyset.
\end{equation}

The identity \eqref{critical} is the basis for work done by
Raymond \cite{r}; see also \cite{DT2,kkl}, which describe the
band structure on the various levels in an inductive way. Let us
recall this result. Following \cite{kkl}, we call a band $I_k
\subset \sigma_k$ a ``type~A band'' if $I_k \subset \sigma_{k-1}$
(and hence $I_k \cap (\sigma_{k+1} \cup \sigma_{k-2}) = \emptyset$).
We call a band $I_k \subset \sigma_k$ a ``type~B band'' if $I_k \subset
\sigma_{k-2}$ (and therefore $I_k \cap \sigma_{k-1} = \emptyset$).
Then, we have the following result (Lemma~5.3 of \cite{kkl},
essentially Lemma~6.1 of \cite{r}):

\begin{lemma}\label{order}
For every $\lambda > 4$ and every $k \ge 1$, \\[1mm]
{\rm (a)} Every type A band $I_k \subset \sigma_k$ contains
exactly one type B band $I_{k+2} \subset \sigma_{k+2}$, and no
other bands from $\sigma_{k+1}$, $\sigma_{k+2}$. \\[1mm]
{\rm (b)} Every type B band $I_k \subset \sigma_k$ contains
exactly one type A band $I_{k+1} \subset \sigma_{k+1}$ and two
type B bands from $\sigma_{k+2}$, positioned around $I_{k+1}$.
\end{lemma}

\begin{lemma}\label{neint}
For every band $I$ of $\sigma_k$, we have that $I \cap
\Sigma_\lambda \not= \emptyset$.
\end{lemma}

\begin{proof}
Let $I$ be a band of $\sigma_k$. Choose a band $I^{(1)}$ in
$\sigma_{k+1} \cup \sigma_{k+2}$ with $I \supset I^{(1)}$,
as is possible by Lemma~\ref{order}.
Iterating this procedure, we obtain
a nested sequence of intervals and hence a point $E \in I$ for
which the corresponding trace map orbit is bounded. Thus, by
\eqref{suto2}, $E \in \Sigma_\lambda$.
\end{proof}

We define
\begin{align*}
a_k & = \text{number of bands of type A in } \sigma_k, \\
b_k & = \text{number of bands of type B in } \sigma_k.
\end{align*}

By Raymond's work, it follows immediately that $a_k + b_k = F_k$
for every $k$. In fact, we have the following result.

\begin{lemma}
The constants $\{a_k\}$ and $\{b_k\}$ obey the relations
\begin{equation}\label{abkrecursion}
a_k = b_{k-1} , \quad b_k = a_{k-2} + 2b_{k-2}
\end{equation}
with initial values
$$
a_0 = 1, \quad a_1 = 0, \quad b_0 = 0, \quad b_1 = 1.
$$
Consequently, for $k \ge 2$,
\begin{equation}\label{abkvalues}
a_k = b_{k-1} = F_{k-2}.
\end{equation}
\end{lemma}

\begin{proof}
The recursions \eqref{abkrecursion} hold by definition. The
explicit expressions in \eqref{abkvalues} then follow quickly by
induction.
\end{proof}

We are interested in the size of the bands of a given type on a
given level. To capture the distribution of these lengths, we
define
\begin{align*}
a_{k,m} & = \text{number of bands $b$ of type A in $\sigma_k$ with } \# \{ 0\le
j < k : b \cap \sigma_j \not= \emptyset \} = m, \\
b_{k,m} & = \text{number of bands $b$ of type B in $\sigma_k$ with
} \# \{ 0 \le j < k : b \cap \sigma_j \not= \emptyset \} = m.
\end{align*}
The motivation for this definition is the following: If we
consider a given band and its location relative to the bands on
the previous levels, each time the given band is contained in a
band on a previous level, we essentially pick up a factor roughly
of size $\lambda^{-1}$ (for $\lambda$ large). Thus, for example,
there are $a_{k,m}$ bands of size $\approx \lambda^{-m}$ of type
$A$ in $\sigma_k$. Since the combinatorics of the situation is
$\lambda$-independent, we choose to separate the two aspects. This
will allow us to get much more precise information in the regime
of large $\lambda$, which is extremely hard to study numerically.

\begin{lemma}
We have
\begin{equation}\label{abkmrecursion}
a_{k,m} = b_{k-1,m-1} , \quad b_{k,m} = a_{k-2,m-1} + 2b_{k-2,m-1}
\end{equation}
with initial values
\begin{align*}
a_{0,m} & = 0 \text{ for } m > 0 \quad \text{ and } \quad a_{0,0} = 1, \\
a_{1,m} & = 0 \text{ for } m \ge 0, \\
b_{0,m} & = 0 \text{ for } m \ge 0, \\
b_{1,m} & = 0 \text{ for } m > 0 \quad \text{ and } \quad b_{1,0}
= 1.
\end{align*}
Consequently,
\begin{equation}\label{abkmzero}
a_{k,m} = b_{k-1,m-1} = \begin{cases} 2^{2k - 3m - 1} \frac{m}{k-m}
{ k - m \choose 2m - k } & \text{ when } \lceil \tfrac{k}{2} \rceil
\le m \le \lfloor \tfrac{2k}{3} \rfloor; \\ 0 & \text{ otherwise.}
\end{cases}
\end{equation}
\end{lemma}

\noindent\textit{Remark.} The fact that $a_{k,m}$ is zero when
$\lceil \tfrac{k}{2} \rceil \le m \le \lfloor \tfrac{2k}{3}
\rfloor$ fails is an immediate consequence of the properties
\eqref{suto1}, \eqref{suto2}, and \eqref{critical} established by
S\"ut\H{o}. As a slight variation of the proof of
Theorem~\ref{lowerboundthm} below shows, this fact alone is
sufficient to give a quick proof of
$$
\dim_B^\pm(\Sigma_\lambda) \ge 1.5 \, \frac{\log \phi}{\log
S_\mathrm{u}(\lambda)}.
$$
Our more detailed description of the numbers $a_{k,m}$ will then
enable us to prove the stronger lower bound with $f^\#$ in place of
$1.5$, which is optimal as discussed in the Introduction.

\begin{proof}
It is immediate from the definition and \eqref{abkrecursion} that
the recursions \eqref{abkmrecursion} hold. Recall that the Chebyshev
polynomials of the first kind are defined by the recurrence relation
\begin{align*}
T_0(x) & = 1, \\
T_1(x) & = x, \\
T_{m+1} (x) & = 2 x T_m(x) - T_{m-1} (x).
\end{align*}
Write $T_m(x)$ as
$$
T_m(x) = \sum_{r=0}^{\lfloor \frac{m}{2} \rfloor} (-1)^r c_{r,m}
x^{m-2r}.
$$
On the one hand, it is known (see, e.g., \cite{as}) that for
$m \ge 1$ and $0 \le r \le \lfloor \frac{m}{2} \rfloor$,
\begin{equation}\label{crmformula}
c_{r,m} = 2^{m-2r-1} \frac{m}{m-r} { m-r \choose r }.
\end{equation}
On the other hand, the recursion generating the polynomials says
that $c_{0,0} = 1$, $c_{0,1} = 1$, and
\begin{align*}
\sum_{r=0}^{\lfloor \frac{m+1}{2} \rfloor} (-1)^r c_{r,m+1} x^{m -
2r + 1} & = 2x \sum_{r=0}^{\lfloor \frac{m}{2} \rfloor} (-1)^r
c_{r,m} x^{m-2r} - \sum_{r=0}^{\lfloor \frac{m-1}{2} \rfloor}
(-1)^{r}
c_{r,m-1} x^{m - 2r - 1} \\
& = \sum_{r=0}^{\lfloor \frac{m}{2} \rfloor} (-1)^r 2 c_{r,m}
x^{m-2r+1} + \sum_{r=0}^{\lfloor \frac{m-1}{2} \rfloor} (-1)^{r+1}
c_{r,m-1} x^{m - 2r - 1} \\
& = \sum_{r=0}^{\lfloor \frac{m}{2} \rfloor} (-1)^r 2 c_{r,m}
x^{m-2r+1} + \sum_{r=1}^{\lfloor \frac{m+1}{2} \rfloor} (-1)^{r}
c_{r-1,m-1} x^{m - 2r + 1}.
\end{align*}
It follows that
\begin{equation}\label{crmrecursion}
c_{r,m+1} = 2 c_{r,m} + c_{r-1,m-1}.
\end{equation}
Denote $\tilde a_{k,m} = c_{2m-k,m}$. Then, using
\eqref{crmrecursion}, we find that
\begin{align*}
\tilde a_{k+1,m+1} & = c_{2(m+1) - (k+1),m+1} \\
& = c_{2m-k+1,m+1} \\
& = 2 c_{2m-k+1,m} + c_{2m-k,m-1} \\
& = 2 \tilde a_{k-1,m} + \tilde a_{k-2,m-1}.
\end{align*}
Comparing this with the recursion
$$
a_{k+1,m+1} =2a_{k-1,m} + a_{k-2,m-1}
$$
we established above, along with the initial values, we see that
\begin{equation}\label{acformula}
a_{k,m} = \tilde a_{k,m} = c_{2m-k,m}.
\end{equation}
Combining \eqref{crmformula} and \eqref{acformula}, we obtain
\eqref{abkmzero}.
\end{proof}

Recall that we introduced the function
\begin{eqnarray*}
   f(x) = \frac{1}{x} \left[\right. \kern-20pt&& (2-3x) \log 2 + (1-x) \log (1-x) \\[-.25em]
          && \kern10pt {}-(2x-1)\log (2x-1) - (2-3x) \log (2-3x) \left.\kern-2pt\right]
\end{eqnarray*}
on the interval $(\frac{1}{2},\frac{2}{3})$. We set
$f(\frac{1}{2}) = \log 2$ and $f(\frac{2}{3}) = 0$ and write
$$
x^* = {12 - 2\sqrt{2} \over 17} = 0.5395042867796\ldots
$$
for the unique point in $(\frac{1}{2},\frac{2}{3})$ where $f$
takes its maximum value,
$$
f^* = f(x^*) = \log(1+\sqrt{2}) = 0.8813735870195\ldots.
$$

\begin{prop}\label{max-log}
For $\frac{k}{2} \le m \le \frac{2k}{3}$, we have
\begin{equation}\label{akmest}
k^{-1/2} \exp\left( m f\left( \frac{m}{k} \right)\right) \lesssim
a_{k,m} \lesssim k^{1/2} \exp\left( m f\left( \frac{m}{k}
\right)\right).
\end{equation}
Consequently,
\begin{equation}\label{limmaxakm}
\lim_{k \to \infty} \max_{m} \frac{1}{m} \log a_{m,k} = f^*.
\end{equation}
\end{prop}

\begin{proof}
As we saw above, when $\frac{k}{2} \le m \le \frac{2k}{3}$,
$$
a_{k,m} = 2^{2k - 3m - 1} \frac{m}{k-m} { k - m \choose 2m - k }.
$$
Thus, we see that $a_{k,\frac{k}{2}} = 2^{\frac{k}{2}-1}$ and
$a_{k,\frac{2k}{3}} = 1$, so the estimate~\eqref{akmest}
holds when $m = \frac{k}{2}$ or $m = \frac{2k}{3}$.

Let us now consider $\frac{k}{2} < m < \frac{2k}{3}$, in which case
\begin{align}
\label{est1} \frac{k}{3}  < & \ k-m < \frac{k}{2}, \\
\label{est2} 1 \le & \ 2m-k < \frac{k}{3}, \\
\label{est3} 1 \le & \ 2k-3m < \frac{k}{2}.
\end{align}
Stirling's approximation gives
$$
n! \asymp \sqrt{n} \left( \frac{n}{e} \right)^n
$$
for every $n \ge 1$, where we write $a\asymp b$ if $a \lesssim b$
and $a \gtrsim b$. Thus,
$$
a_{k,m} \asymp 2^{2k-3m} \frac{m}{k-m} \left(
\frac{k-m}{(2m-k)(2k-3m)} \right)^{1/2}
\frac{\left(\frac{k-m}{e}\right)^{k-m} }{ \left(\frac{2m-k}{e}
\right)^{2m-k} \left(\frac{2k-3m}{e} \right)^{2k-3m}}.
$$
Due to the estimates \eqref{est1}--\eqref{est3}, \eqref{akmest} will
follow once we show that
$$
\tilde a_{k,m} := 2^{2k-3m} \frac{\left(\frac{k-m}{e}\right)^{k-m}
}{ \left(\frac{2m-k}{e} \right)^{2m-k} \left(\frac{2k-3m}{e}
\right)^{2k-3m}} = \exp\left( m f\left( \frac{m}{k} \right)\right).
$$
Writing $x = m/k$, we find that
\begin{align*}
\tilde a_{k,m} & = 2^{2k-3m} \frac{\left(k-m\right)^{k-m} }{
\left(2m-k\right)^{2m-k} \left(2k-3m \right)^{2k-3m}} \\
& = 2^{2k-3kx} \frac{\left(k-kx\right)^{k-kx} }{
\left(2kx-k \right)^{2kx-k} \left(2k-3kx \right)^{2k-3kx}} \\
& = 2^{2k-3kx} \frac{\left(1-x\right)^{k(1-x)} }{ \left(2x-1
\right)^{k(2x-1)} \left(2-3x \right)^{k(2-3x)}} \\
& = \exp\left( k x f\left( x \right)\right) \\
& = \exp\left( m f\left( \frac{m}{k} \right)\right).
\end{align*}
This finishes the proof of \eqref{akmest}, which in turn immediately
implies
$$
\limsup_{k \to \infty} \max_{m} \frac{1}{m} \log a_{m,k} \le f^*.
$$
On the other hand, as $k$ gets large, we can choose $m$ so that
$m/k$ gets arbitrarily close to $x^*$, so that by
\eqref{akmest} again,
$$
\liminf_{k \to \infty} \max_{m} \frac{1}{m} \log a_{m,k} \ge f^*.
$$
Thus, we have established~\eqref{limmaxakm}.
\end{proof}

\section{A Lower Bound for the Box Counting Dimension of the
Spectrum}\label{boxsec}

In this section we will prove a lower bound for the (lower) box
counting dimension of the spectrum of the Fibonacci Hamiltonian.
Recall that the lower box counting dimension of a bounded set $S
\subset \R$ is defined as follows:
$$
\dim_B^-(S) = \liminf_{\varepsilon \to 0} \frac{\log
N_S(\varepsilon)}{\log 1/\varepsilon},
$$
where
$$
N_S(\varepsilon) = \# \{ j \in \Z : [j\varepsilon,
(j+1)\varepsilon) \cap S \not= \emptyset \}.
$$

The following was shown in \cite{DT}.
\begin{lemma}\label{dtlemma}
For every $\lambda > 4$, there exists $S_u(\lambda)$ such that the following
holds: \\[1mm]
{\rm (a)} Given any {\rm (}type A{\rm )} band $I_{k+1} \subset
\sigma_{k+1}$ lying in the band $I_k \subset \sigma_k$, we have
for every $E \in I_{k+1}$,
$$
\left| \frac{x_{k+1}'(E)}{x_k'(E)} \right| \le S_u(\lambda).
$$
{\rm (b)} Given any {\rm (}type B{\rm )} band $I_{k+2} \subset
\sigma_{k+2}$ lying in the band $I_k \subset \sigma_k$, we have
for every $E \in I_{k+2}$,
$$
\left| \frac{x_{k+2}'(E)}{x_k'(E)} \right| \le S_u(\lambda).
$$
For example, one can choose $S_u(\lambda) = 2\lambda + 22$.
\end{lemma}

We use the symbol $S_u(\lambda)$ instead of the explicit $2\lambda
+ 22$ to make the dependence of everything that follows on this
quantity explicit.

Inductively, Lemma~\ref{dtlemma} gives an upper bound for $|x_k'|$
on each band of $\sigma_k$ and this in turn will give a lower
bound on the length of a band, which will be crucial in the proof
of the following result.

\begin{theorem}\label{lowerboundthm}
For all $\lambda > 4$,
\begin{equation}\label{dimblowerexpl}
\dim_B^-(\Sigma_\lambda) \ge \frac{f^*}{\log S_u(\lambda)}.
\end{equation}
\end{theorem}

\begin{proof}
Let $m_k = \lfloor 3kx^* \rfloor$. Since $x^*>0.4$, the sequence
$\{m_k\}$ is strictly increasing and $\lim_{k \to \infty}
\frac{m_k}{3k} = x^*$. Write
$$
f_k = \frac{1}{m_k} \log a_{3k,m_k}.
$$
We know from Proposition~\ref{max-log} that
\begin{equation}\label{fkconv}
\lim_{k \to \infty} f_k = f^*
\end{equation}
and it is obvious that
\begin{equation}\label{mkconv}
\lim_{k \to \infty} \frac{m_{k+1}}{m_{k}} = 1.
\end{equation}
For a given $k$, let us consider the type A bands in $\sigma_{3k}$
that lie in $m_k$ bands on previous levels. By definition, there
are $N_k:= a_{3k,m_k}$ such bands. By Lemma~\ref{dtlemma}, each of
them has length at least $\varepsilon_k := 4 S_u (\lambda)^{-m_k}$,
since on each band of $\sigma_{3k}$, the function $x_{3k}$ is
strictly monotone and runs from $\pm 2$ to $\mp 2$.  Let
$\{A_{3k,j} \}_{j=1}^{N_k}$ be these bands, enumerated so that
$A_{3k,j}$ is to the left of $A_{3k,j+1}$ for every $j$. Each band
has non-empty intersection with $\Sigma_\lambda$ by
Lemma~\ref{neint}. Thus, for every $j$, there exists $E_{3k,j} \in
A_{3k,j} \cap \Sigma_\lambda$. Clearly, the energies $\{ E_{3k,j}
\}$ are increasing in $j$. Consider $\{ E_{3k,j} \}$ with $j$ odd,
that is, $j = 2s+1$, $0 \le s \le \lfloor N_k/2 \rfloor$. Since
any two bands $A_{3k, 2s-1}$, $A_{3k, 2s+1}$ are separated by the
band $A_{3k,2s}$, which has length at least $\varepsilon_k$, we
get $|x_{3k, 2s-1}-x_{3k, 2s+1}| \ge \varepsilon_k$ for every $s$.
Thus, the $E_{3k,2s+1}$ belong to different $\varepsilon$-boxes if
$\varepsilon < \varepsilon_k$. We conclude that
$N_{\Sigma_\lambda} (\varepsilon) \ge \frac{N_k}{2}$ for every
$\varepsilon <\varepsilon_k$.

Given any $\varepsilon > 0$, choose $k$ with $\varepsilon_{k+1} \le
\varepsilon < \varepsilon_{k}$. Then,
$$
\frac{\log N_{\Sigma_\lambda}(\varepsilon)}{\log 1/\varepsilon} \ge \frac{\log (a_{3k,m_k}) - \log 2 }{\log
1/\varepsilon_{k+1}} \ge \frac{f_k - \frac{1}{m_k}\log
2}{\frac{m_{k+1}}{m_k}\log S_u(\lambda)}.
$$
Since, as $\varepsilon \to 0$, we have $k,m_k \to \infty$ and
$$
\lim_{k \to \infty} \frac{f_k}{\frac{m_{k+1}}{m_k}\log S_u(\lambda)}
= \frac{f^*}{\log S_u(\lambda)},
$$
by \eqref{fkconv} and \eqref{mkconv}, the result follows.
\end{proof}

\section{An Upper Bound for the Hausdorff Dimension of the
Spectrum}\label{haussec}

In this section we prove an upper bound for the Hausdorff
dimension of $\Sigma_\lambda$. We will use the canonical coverings
$\sigma_k \cup \sigma_{k+1}$ of $\Sigma_\lambda$ and estimate the
lengths of the intervals in these periodic spectra from above
using our combinatorial result from Section~\ref{bandsec} together
with a scaling result that is analogous to Lemma~\ref{dtlemma},
but which gives bounds from the other side.

Namely, the following was shown by Killip, Kiselev, and Last
\cite[Lemma~5.5]{kkl}.
\begin{lemma}\label{kkllemma}
For every $\lambda \ge 8$, there exists $S_l(\lambda)$ such that the following
holds: \\[1mm]
{\rm (a)} Given any {\rm (}type A{\rm )} band $I_{k+1} \subset
\sigma_{k+1}$ lying in the band $I_k \subset \sigma_k$, we have for
every $E \in I_{k+1}$,
$$
\left| \frac{x_{k+1}'(E)}{x_k'(E)} \right| \ge S_l(\lambda).
$$
{\rm (b)} Given any {\rm (}type B{\rm )} band $I_{k+2} \subset
\sigma_{k+2}$ lying in the band $I_k \subset \sigma_k$, we have for
every $E \in I_{k+2}$,
$$
\left| \frac{x_{k+2}'(E)}{x_k'(E)} \right| \ge S_l(\lambda).
$$
For example, one can choose
$$
S_l(\lambda) = \frac{1}{2} \left( (\lambda - 4) + \sqrt{(\lambda -
4)^2 - 12} \right).
$$
\end{lemma}

As before, we use the symbol $S_l(\lambda)$ instead of the
explicit possible choice to make the dependence of everything
that follows on this quantity explicit.

Inductively, Lemma~\ref{kkllemma} gives a lower bound for $|x_k'|$
on each band of $\sigma_k$ and hence an upper bound for the length
of a band.

\begin{theorem}\label{thm5}
Suppose $\lambda \ge 8$. Then
$$
\dim_H (\Sigma_\lambda) \le \frac{f^*}{\log S_l(\lambda)}.
$$
\end{theorem}

\noindent\textit{Remark.} Raymond \cite{r} proved an upper bound
of the form
$$
\dim_H (\Sigma_\lambda) \le \frac{2 \log \phi}{\log S_l(\lambda)}.
$$
As above, to get this weaker result, all one needs is information
about the support of $a_{k,m}$, which in turn follows quickly from
the properties \eqref{suto1}, \eqref{suto2}, and \eqref{critical}.
Using our more detailed information about the values of $a_{k,m}$ on
the support, we can improve the constant $2$ to the value $f^\#
\approx 1.83157$, which is optimal by our discussion in the
introduction. Our proof is inspired by Raymond's proof, and the
improvement stems from our more detailed analysis of the
distributions of bandwidths in the approximating periodic spectra.

\begin{proof}
As recalled above, the set $\sigma_k \cup \sigma_{k+1}$ is a finite
union of compact real intervals that covers $\Sigma_\lambda$. There
are $a_{k,m} + b_{k,m}$ bands in $\sigma_k$ that lie in exactly $m$
intervals on previous levels. We know that the length of each of these
bands is bounded from above by $4 S_l(\lambda)^{-m}$.

Thus, it suffices to show that, given any
\begin{equation}\label{chooses}
s > \frac{f^*}{\log S_l(\lambda)},
\end{equation}
we have
$$
\lim_{k \to \infty} \sum_{m = \lceil \frac{k}{2} \rceil}^{\lfloor
\frac{2k}{3} \rfloor} \left( a_{k,m} + b_{k,m} \right) (4
S_l(\lambda))^{-sm} + \sum_{m = \lceil \frac{k+1}{2}
\rceil}^{\lfloor \frac{2(k+1)}{3} \rfloor} \left( a_{k+1,m} +
b_{k+1,m} \right) (4 S_l(\lambda))^{-sm} = 0.
$$
For simplicity, we will only consider
$$
C_k = \sum_{m = \lceil \frac{k}{2} \rceil}^{\lfloor \frac{2k}{3}
\rfloor} a_{k,m} S_l(\lambda)^{-sm};
$$
the other terms can be dealt with in an analogous way.

By our upper bound for $a_{k,m}$ established earlier and the
assumption \eqref{chooses}, we have
\begin{align*}
C_k & \lesssim k^{1/2} \sum_{m = \lceil \frac{k}{2} \rceil}^{\lfloor
\frac{2k}{3} \rfloor} \exp \left( m (f(\tfrac{m}{k}) - s \log
S_l(\lambda))\right) \\
& \lesssim k^{3/2} \exp \left( \frac{k}{2} (f^* - s \log
S_l(\lambda))\right).
\end{align*}
Using \eqref{chooses} again, we see that the right-hand side goes to
zero as $k \to \infty$.
\end{proof}

Theorem~\ref{lowerboundthm} and Theorem~\ref{thm5} together
establish Theorem~\ref{mainthm} from the Introduction.

\section{Transfer Matrices, Box Counting Dimension, and Wavepacket
Spreading}\label{qdsec}

In this section we will first consider general half-line
Schr\"odinger operators and derive a number of consequences from
polynomially bounded transfer matrices. We then turn to general
Schr\"odinger operators on the line and use our half-line results to
derive a lower bound for the exponent governing the spreading of the
fastest part of the wavepacket in terms of the box counting
dimension of the spectrum. These results are relevant in our context
since for the Fibonacci potential, it is known that the transfer
matrices are polynomially bounded for all parameter values and all
energies in the spectrum.

Consider a discrete half-line Schr\"odinger operator on
$\ell^2(\Z_+)$,
\begin{equation}\label{box001}
[H_+ u](n)=u(n+1) + u(n-1) + V(n) u (n),
\end{equation}
with Dirichlet boundary condition $u(0) = 0$. Here, $V$ is a
bounded real-valued function. $H_+$ is a bounded self-adjoint
operator in $\ell^2(\Z_+)$. The vector $\delta_1$ is cyclic and we
denote its spectral measure by $\mu$.

For $z \in \C$ and $m,n \in \Z_+$, we denote by $T(n,m; z)$ the
transfer matrix associated with the difference equation
\begin{equation}\label{eve}
u(n+1) + u(n-1) + V(n) u(n) = z u(n).
\end{equation}
That is, $T(n,m;z)$ is the unique unimodular $2 \times 2$ matrix
for which we have
$$
\begin{pmatrix} u(n+1) \\ u(n) \end{pmatrix} = T(n,m;z) \begin{pmatrix}
u(m+1) \\ u(m) \end{pmatrix}
$$
for every solution of \eqref{eve}.

To study the spreading of a wavepacket with initial state
$\delta_1$ under the dynamics generated by $H$, one usually
considers for $p > 0$,
$$
\langle |X|_{\delta_1}^p \rangle (T) = \frac{2}{T} \int_0^{\infty}
e^{-2t/T} \sum_{n \in \Z_+} n^p | \langle e^{-itH_+} \delta_1,
\delta_n \rangle |^2 \, dt.
$$
One is interested in the power-law growth of this quantity and
therefore studies  the lower transport exponent
$$
\beta^-_{\delta_1}(p)=\liminf_{T \to \infty} \frac{\log \langle
|X|_{\delta_1}^p \rangle (T) }{p \, \log T}
$$
and the upper transport exponent
$$
\beta^+_{\delta_1}(p)=\limsup_{T \to \infty} \frac{\log \langle
|X|_{\delta_1}^p \rangle (T) }{p \, \log T}.
$$
Both functions $\beta^\pm_{\delta_1} (p)$ are nondecreasing in $p$
and hence the following limits exist:
$$
\alpha_u^\pm = \lim_{p \to \infty} \beta_{\delta_1}^\pm (p).
$$
Alternatively (see \cite[Theorem 4.1]{GKT}),
we also have $\alpha_u^\pm = \sup \{ \alpha \ge 0 :
S^\pm (\alpha) < \infty \}$, where
\begin{align*}
S^-(\alpha) & = - \liminf_{T \to \infty} \frac{\log P_+(T^\alpha -2,
T) }{\log T} \\
S^+(\alpha) & = - \limsup_{T \to \infty} \frac{\log P_+(T^\alpha -2,
T) }{\log T}
\end{align*}
and
$$
P_+(N,T) = \sum_{n > N} \frac{2}{T} \int_0^{\infty} e^{-2t/T} |
\langle e^{-itH_+} \delta_1, \delta_n \rangle |^2 \, dt.
$$

The following result derives several consequences from the fact
that the transfer matrices are polynomially bounded for energies
from some subset of the real line. We denote by $\mu$ the spectral
measure associated with the vector $\delta_1$.

\begin{theorem}\label{box}\label{thmpowerlaw}
Consider a half-line Schr\"odinger operator $H_+$ with bounded
potential as above. Let $A \subset [-B,B]$, $B>0$, and assume that
there exist positive constants $C, \alpha$ such that for every $E
\in A$ and every $N \ge 1$,
\begin{equation}\label{box101}
\|T(n,m; E)\| \le C N^\alpha \text{ for all } m,n \in [1,N].
\end{equation}
{\rm (a)} For every $\sigma>0$, there exists $\e_0(\sigma) > 0$
such that
\begin{equation}\label{box01}
\mu([E-\e, E+\e]) \ge D \e^L
\end{equation}
for $\e \in (0,\e_0)$ and $E \in A$. Here, $L=(1+\sigma)(1+3
\alpha)$ and $D$ depends only on $\sigma$.
\\[1mm]
{\rm (b)} For any $p>0$, we have that
\begin{equation}\label{box02}
\beta_{\delta_1}^\pm (p) \ge \left( 1 + \frac{1}{p} \right)
\dim_B^\pm (A) - \frac{1+3 \alpha}{p}.
\end{equation}
In particular,
\begin{equation}\label{box03}
\alpha_u^\pm \ge {\rm dim}_B^\pm (A).
\end{equation}
\end{theorem}

\begin{proof}
It was shown by Germinet, Kiselev and Tcheremchantsev
\cite[Proposition 2.1]{GKT} that for any $M>0$,
\begin{equation}\label{box1}
\mu([E-\e, E+\e])\ge C_1 \int_{E-\e/2}^{E+\e/2} \|T(N,0; x)\|^{-2}
\, dx - C_2 \e^M.
\end{equation}
Here $\e \in (0,1)$, $E \in [-B,B]$, $N = \lfloor \e^{-1-\sigma}
\rfloor$, $\sigma>0$, and the constants $C_1, C_2$ depend only on
$B,M,\sigma$.

Let $E \in A$. To bound the integral in \eqref{box1} from below,
we need an upper bound for $\|T(N,1; x)\|$ for $x$ close to $E$.
This can be accomplished using a perturbative argument of Simon;
compare \cite[Lemma~2.1]{DT}. Namely, if
$$
D(N,E) = \sup_{1 \le n, m \le N} \|T(n,m; E) \|,
$$
then for any $\delta \in \C$ and $1 \le n \le N$ we have
$$
\| T(n,1; E+\delta) \| \le D(N,E) \exp (D(N,E) |n|\delta).
$$
By assumption, $D(N,E) \le CN^\alpha$ for $E \in A$ and $N \ge 1$.
Thus,
$$
\| T(N,0; E+\delta) \| \le C N^\alpha \exp(CN^{1+\alpha} \delta),
$$
and for any $\delta$ with $|\delta| \le \gamma=1/2N^{-1-\alpha}$, we
get
\begin{equation}\label{box2}
\| T(N,1; E+\delta) \| \le C' N^\alpha
\end{equation}
with a uniform constant.

Since $N = \lfloor \e^{-1-\sigma} \rfloor$, $\e \in (0,1)$, we see
that $\gamma<\e/2$ for $\e$ small enough and therefore $[E-\gamma,
E+\gamma] \subset [E-\e/2, E+\e/2]$. Using \eqref{box1} and
\eqref{box2}, we find
\begin{align*}
\mu([E-\e, E+\e]) & \ge C_1 \int_{E-\gamma}^{E+\gamma} \|T(N,0;
x)\|^{-2} dx - C_2 \e^M \\
& \ge C_3 \gamma N^{-2\alpha} - C_2 \e^M \\
& = C_4 N^{-1-3\alpha} - C_2 \e^M \\
& = C_4 \e^{(1+\sigma)(1+3 \alpha)} -C_2\e^M
\end{align*}
with appropriate constants. Taking $M=(1+\sigma)(1+3 \alpha)+1$,
we obtain \eqref{box01}.

In order to prove the dynamical lower bound \eqref{box02}, we will
apply \cite[Corollary~4.1]{T2}; see also \cite{BGT}. In our
notation, this result reads
\begin{equation}\label{box31}
\beta_{\delta_1}^\pm (p) \ge D_\mu^\pm (q), \  q \in \Big(
\frac{1}{1+p}, 1 \Big),
\end{equation} where $D_\mu (q), \ 0< q<1$ are the generalized
fractal dimensions of the spectral measure; compare \cite{BGT,
BGT2}. One of the equivalent ways of defining these dimensions is
the following (cf.~\cite[Theorem 4.3]{T2}):
\begin{equation}\label{box4}
D_\mu^+ (q) = \limsup_{\e \to 0} \frac {\log S_\mu (q, \e)} {(q-1)
\log \e},
\end{equation}
and similarly for $D_\mu^-(q)$ with $\limsup$ replaced by
$\liminf$, where
$$
S_\mu (q, \e)=\sum_{j \in \Z} (\mu ([j \e, (j+1)\e))^q.
$$

To bound $S_\mu (q,\e)$ from below, we can follow the proof of
\cite[Theorem 4.5]{T2}. Let $I_j=[j\e, (j+1)\e)$. Denote by $J_A$
the set of $j \in \Z$ for which $I_j \cap A \ne \emptyset$. Thus,
for every $j \in J_A$, there exists $E_j \in [j\e, (j+1)\e)$ such
that $E_j \in A$. The bound \eqref{box01} implies for $\e$ small
enough,
$$
\mu(I_{j-1})+\mu(I_j)+\mu (I_{j+1}) \ge \mu([E_j-\e, E_j+\e]) \ge
D \e^L,
$$
for every $j \in J_A$, where $L=(1+\sigma)(1+3\alpha)$. Thus,
$$
B \equiv \sum_{j \in \Z} \left( \mu(I_{j-1})+\mu(I_j)+\mu
(I_{j+1}) \right)^q \ge D^q \e^{qL} \sum_{j \in J_A} 1.
$$
On the other hand, since $(a+b)^q \le a^q+b^q, \ a,b>0, \ q \in
(0,1)$, one has $B \le 3 S_\mu (q, \e)$. Therefore,
\begin{equation}\label{box5}
S_\mu (q,\e) \ge C(q) \e^{qL} \sum_{j \in J_A} 1 \equiv C(q) \e^{qL}
N(\e)
\end{equation}
with
$$
N(\e)=\#  \{ j \in \Z \ : \ I_j \cap A \ne \emptyset
\}.
$$
It follows from \eqref{box4}--\eqref{box5} that
\begin{equation}\label{box7}
D_\mu^\pm (q) \ge \frac {\dim_B^\pm (A)-qL}  {1-q}, \quad q \in (0,1).
\end{equation}
The bounds \eqref{box31} and \eqref{box7} imply, for any $q \in
(\frac{1}{1+p}, 1)$,
$$
\beta_{\delta_1}^\pm (p) \ge \frac{\dim_B^\pm (A)- qL} {1-q}.
$$
Since this holds for $L=(1+\sigma)(1+3 \alpha)$ with any $\sigma>0$,
letting $\sigma \to 0$ and $q \to \frac{1}{1+p}$, we obtain
\eqref{box02}.

Since
$$
\alpha_u^\pm=\lim_{p \to +\infty} \beta_{\delta_1}^\pm (p),
$$
the estimate \eqref{box03} follows.
\end{proof}

\noindent\textit{Remark.} One can improve the bounds \eqref{box01},
\eqref{box02} if one has a nontrivial (i.e., better than ballistic)
upper bound on $\alpha_u^+$. In this case \cite{T3}, \eqref{box1}
holds with $N = \lfloor \e^{-\rho} \rfloor$, where $\rho =
\alpha_u^+ + \sigma$, $\sigma>0$. Thus, \eqref{box01} holds with a
smaller value of $L$.

\bigskip

Let us now turn to the whole-line case and prove a result
analogous to \eqref{box03}. Consider the operator $H$ acting on
$\ell^2(\Z)$ as
$$
[Hu](n) = u(n+1) + u(n-1) + V(n)u(n)
$$
with bounded $V : \Z \to \R$, and choose $K$ so that $\sigma (H)
\subset [-K+1,K-1]$. The transfer matrices $T(n,m;z)$ are defined
in an analogous way.

For $N \ge 1$, recall that the time-averaged right outside
probabilities are given by
$$
P_r(N,T)= \frac{2}{T} \int_0^\infty e^{-\frac{2t}{T}}
\sum_{n > N} \left| \left\langle e^{-itH} \delta_1, \delta_n
\right\rangle \right|^2 \, dt.
$$
It was proved in \cite{DT2} that
\begin{equation}\label{box8}
C T^{-3} I_r (N,T)\le P_r (N,T)\le C e^{-cN} + C T^3 I_r (N,T),
\end{equation} where
$$
I_r (N,T)=\int_{-K}^K \| T (N,1; E + \tfrac{i}{T} ) \|^{-2} \, dE.
$$
Although not stated explicitly, this result from \cite{DT2} can be
extended easily to the half-line case \eqref{box001}. Thus, for
the time-averaged outside probabilities $P_+(N,T)$, we have
\begin{equation}\label{box9}
C T^{-3} I_+ (N,T)\le P_+(N,T)\le C e^{-cN} + C T^3 I_+ (N,T),
\end{equation}
where
$$
I_+(N,T) = \int_{-K}^K \|T(N,1; E+ \tfrac{i}{T} )\|^{-2} \, dE.
$$

\begin{coro}\label{full-line}
Suppose that $H$ is a discrete Schr\"odinger operator on the line
with bounded potential. Assume $A \subset [-B,B] \subset \R$ is
such that for $E \in A$ and $N \ge 1$,
\begin{equation}\label{box102}
\|T(n,m; E)\| \le C N^\alpha \text{ for all } m,n \in [1,N]
\end{equation}
with uniform constants $C,\alpha$. Then, for the initial state
$\delta_1$, we have that
$$
\alpha_u^\pm \ge \dim_B^\pm (A).
$$
\end{coro}

\begin{proof}
Let $H_+$ be the half-line operator of the form \eqref{box001} with
potential $V_+$ given by $V_+(n)=V(n)$ for $n \ge 1$. The transfer
matrices are the same for both operators since they are defined
locally, and thus we have $I_r(N,T) = I_+(N,T)$. The condition
\eqref{box101} holds for the operator $H_+$ and Theorem~\ref{box}
therefore yields
$$
\alpha_{u,+}^\pm \ge \dim_B^\pm (A)
$$
for the corresponding dynamics.

Let us assume that $\dim_B^-(A) > 0$ (if $\dim_B^-(A) = 0$, the
result is trivially true) and show that $\alpha_u^- \ge
\dim_B^-(A)$. Take any $\alpha $ with $0 < \alpha < \dim_B^-(A)
\le \alpha_{u,+}^-$. Due to the definition of $\alpha_{u,+}^\pm$,
we have that
\begin{equation}\label{box10}
P_+ (T^\alpha - 2, T) \ge T^{-M}
\end{equation}
with some finite $M > 0$ for $T$ sufficiently large. Since
$I_r(N,T)=I_+(N,T)$, it follows from \eqref{box8} and \eqref{box9}
that
$$
P_r (N,T) \ge C T^{-3} I_+(N,T) \ge C T^{-6} \left( P_+(N,T)-C
e^{-cN} \right).
$$
It follows from \eqref{box10} that for $T$ sufficiently large,
$$
P_r (T^\alpha-2,T) \ge  C T^{-M-6}.
$$
Consequently, for the full time-averaged outside probabilities, we
find
\begin{align*}
P (T^\alpha -2 ,T) & = \frac{2}{T} \int_0^\infty e^{-\frac{2t}{T}}
\sum_{|n|>T^\alpha} \left| \left\langle e^{-itH} \delta_1,
\delta_n \right\rangle \right|^2 \, dt \\
& \ge C T^{-M-6}.
\end{align*}
It now follows directly from the definition of $\alpha_u^-$ that
$\alpha_u^- \ge \alpha$. Since this holds for every $\alpha <
\dim_B^-(A)$, the result follows. For $\alpha_u^+$, the proof is
the completely analogous. (The bound \eqref{box10} then holds for
some sequence of times.)
\end{proof}

\noindent\textit{Remarks.} (a) One can prove a similar result under
the condition
$$
\|T(n,m; E)\| \le C N^\alpha \text{ for all } m,n \in [-N,-1],
$$
where $E \in A, \ N \ge 1$.
\\[2mm]
(b) This establishes Theorem~\ref{thm3} as stated in the
Introduction since it has been shown that in the Fibonacci case,
the transfer matrices are polynomially bounded for all energies in
the spectrum with uniform constants $C, \alpha$.
 See Iochum-Testard \cite{it} for the case $\theta = 0$ and
Damanik-Lenz \cite{dl2} for the case of general $\theta$.

\begin{appendix}

\section{The Hyperbolicity of the Trace Map at Large Coupling and
Some of its Consequences}

\subsection{Description of the Trace Map} The main tool that we are
using here is the \textit{trace map}. It was originally introduced
in \cite{K,kkt}; see also \cite{su} for proofs of the results
described below. Let us recall that the numbers $x_k =
x_k(E,\lambda)$ introduced in Section~\ref{bandsec} satisfy the
recursion relation
\begin{equation}\label{apptm}
x_{k+1}=x_k x_{k-1} - x_{k-2},
\end{equation}
with initial conditions $x_{-1} = 2$, $x_0 = E$, $x_1 = E -
\lambda$, and the invariance relation
\begin{equation}\label{appinv}
x_{k+1}^2 + x_k^2 + x_{k-1}^2 - x_{k+1} x_k x_{k-1} = 4+\lambda^2
\end{equation}
for every $k\in \Z_+$. Because of \eqref{apptm}, it is natural to
consider the so-called \textit{trace map},
$$
T : \R^3 \to \R^3, \; T(x,y,z) = (xy-z,x,y).
$$
The sequence $\left\{x_1, x_2, x_3, \ldots \right\}$ can be
considered as the sequence of first coordinates of points in the
trace map orbit having initial condition $\left( x_1, x_0,
x_{-1}\right)$.

By \eqref{appinv}, the following function is invariant under the
action of $T$:
$$
I = x^2 + y^2 + z^2 - x y z - 4.
$$
In other words, $T$ preserves the family of cubic surfaces 
$$
\mathbb{S}_{I}=\{(x,y,z) \in \R^3 : x^2 + y^2 +z^2 - xyz - 4 = I \}.
$$
The surface $\mathbb{S}_0$ is called the \textit{Cayley cubic}. Denote by
$l_{\lambda}$ the line
$$
l_{\lambda} = \{ (E-\lambda, E, 2) : E \in \R \}.
$$
It is easy to check that $l_{\lambda} \subset \mathbb{S}_{\lambda^2}$.

The following result, proved in \cite{su}, characterizes the
spectrum of the Fibonacci Hamiltonian in terms of trace map dynamics
and therefore establishes an important and fruitful connection
between spectral and dynamical issues in this context.

\begin{theorem}\label{spectrum}
The energy $E$ belongs to the spectrum $\Sigma_\lambda$ of the
Fibonacci Hamiltonian if and only if the positive semiorbit of the
point $(E-\lambda, E, 2)$ under iterates of the trace map $T$ is
bounded.
\end{theorem}

\subsection{Hyperbolicity of the Trace Map for Large \boldmath$\lambda$}

Denote by $f_{\lambda}$ the restriction of the trace map to the invariant
surface $\mathbb{S}_{\lambda^2}$. That
is, $f_{\lambda} : \mathbb{S}_{\lambda^2} \to \mathbb{S}_{\lambda^2}$ and
$f_{\lambda} =
T|_{\mathbb{S}_{\lambda^2}}$. Denote by $\Omega_{\lambda}$ the set of points in
$\mathbb{S}_{\lambda^2}$ whose
full orbits under $f_{\lambda}$ are bounded.

Let us recall that an invariant set $\Lambda$ of a diffeomorphism
$f:M\to M$ is {\it locally maximal} if there exists a neighborhood
$U(\Lambda)$ such that
$$
\Lambda=\bigcap_{n \in \Z}f^n(U).
$$
An invariant closed set $\Lambda$ of a diffeomorphism $f : M \to
M$ is \textit{hyperbolic} if there exists a splitting of the
tangent space $T_x M = E^s_x \oplus E^u_x$ at every point $x\in
\Lambda$ such that it is invariant under $Df$, and $Df$
exponentially contracts vectors from the stable subspaces
$\{E^s_x\}$ and exponentially expands vectors from the unstable
subspaces $\{E^u_x\}$.

See \cite{H} for a detailed survey of hyperbolic dynamics and an
extensive list of references.

Casadgli proved the following result \cite{C}.

\begin{theorem}\label{Casdagli}
For every $\lambda \ge 16$, the set $\Omega_{\lambda}$ is a locally maximal
invariant hyperbolic set of
$f_{\lambda} : \mathbb{S}_{\lambda^2} \to \mathbb{S}_{\lambda^2}$.
\end{theorem}

\subsection{Some Properties of Locally Maximal Hyperbolic Invariant
Sets of Surface Diffeomorphisms}

Consider a locally maximal invariant transitive hyperbolic set
$\Lambda \subset M$, $\dim M = 2$, of a diffeomorphism $f \in
\mathrm{Diff}^r(M)$, $r\ge 1$. We have $\Lambda = \cap_{n \in \Z}\,
f^n (U(\Lambda))$ for some neighborhood $U(\Lambda)$. Assume also
that $\dim E^u = \dim E^s = 1$.

Let us gather several results in this general context that we will eventually
specialize to the case where
$\Lambda$ is given by $\Omega_{\lambda}$ for $f_{\lambda} : \mathbb{S}_{\lambda^2}
\to \mathbb{S}_{\lambda^2}$,
$\lambda\ge 16$.
\\[2mm]
\textbf{I. Stability.} There is a neighborhood $\mathcal{U}\subset
\mathrm{Diff}^1(M)$ of the map $f$ such that for every $g\in
\mathcal{U}$, the set $\Lambda_g = \cap_{n\in \Z}\, g(U(\Lambda))$ is
a locally maximal invariant hyperbolic set of $g$. Moreover,
there is a homeomorphism $h : \Lambda \to \Lambda_g$ that conjugates
$f|_{\Lambda}$ and $g|_{\Lambda_g}$, that is, the following diagram
commutes:
$$
\CD \Lambda @>{f|_{\Lambda}}>>\Lambda \\
@V{h}VV@VV{h}V\\
\Lambda_g @>g|_{\Lambda_g}>>\Lambda_g\endCD
$$
\textbf{II. Invariant Manifolds.} For $x\in \Lambda$ and small
$\varepsilon > 0$, consider the local stable and unstable sets
$$
W^s_{\varepsilon}(x) = \{w \in M :  d(f^n(x), f^n(w)) \le
\varepsilon \ \ \text{ for all}\ \  n\ge 0\},
$$
$$
W^u_{\varepsilon}(x) = \{ w \in M :  d(f^n(x), f^n(w))\le
\varepsilon \ \ \text{\rm for all}\ \  n\le 0\}.
$$
If $\varepsilon>0$ is small enough, then these are embedded
$C^r$-disks with $T_x W^s_{\varepsilon}(x) = E^s_x$ and $T_x
W^u_{\varepsilon} (x) = E^u_x$. Define the (global) stable and
unstable sets as
$$
W^s(x)=\cup_{n\in \Z_+}f^{-n}(W^s_{\varepsilon}(x)), \ \ \
W^u(x)=\cup_{n\in \Z_+}f^{n}(W^u_{\varepsilon}(x)).
$$
Define also
$$
W^s(\Lambda)=\cup_{x\in \Lambda}W^s(x)\ \ \ \text{\rm and}\ \ \
W^u(\Lambda)=\cup_{x\in \Lambda}W^u(x).
$$
\textbf{III. Invariant Foliations.} A stable foliation for $\Lambda$
is a foliation $\mathcal{F}^s$ of a neighborhood of $\Lambda$ such
that
\begin{itemize}
\item[(a)] for each $x\in \Lambda$, $\mathcal{F}(x)$, the leaf
containing $x$, is tangent to $E^s_x$;
\item[(b)] for each $x\in
\Lambda$, sufficiently near $\Lambda$, $f(\mathcal{F}^s(x))\subset
\mathcal{F}^s(f(x))$.
\end{itemize}
An unstable foliation $\mathcal{F}^u$ can be defined in a similar
way.

For a locally maximal hyperbolic set $\Lambda\subset M$ of a
$C^1$-diffeomorphism $f : M \to M$, $\dim M = 2$, stable and
unstable $C^0$ foliations with $C^1$-leaves can be constructed
\cite{M}. In the case of $C^2$-diffeomorphisms,  $C^1$ invariant
foliations exist (see \cite{PT}, Theorem 8 in Appendix 1).
\\[2mm]
\textbf{IV. Local Hausdorff Dimension and Box Counting Dimension.}
Consider, for $x\in \Lambda$ and small $\varepsilon>0$, the set
$W^u_{\varepsilon}(x)\cap \Lambda$. The Hausdorff dimension of this
set does not depend on $x\in \Lambda$ and $\varepsilon>0$,
and coincides with its box counting dimension (see \cite{MM,T}):
$$
\text{\rm dim}_HW^u_{\varepsilon}(x)\cap \Lambda= \text{\rm
dim}_BW^u_{\varepsilon}(x)\cap \Lambda.
$$
In a similar way,
$$
\text{\rm dim}_HW^s_{\varepsilon}(x)\cap \Lambda= \text{\rm
dim}_BW^s_{\varepsilon}(x)\cap \Lambda.
$$
Denote $h^s=\text{\rm dim}_HW^s_{\varepsilon}(x)\cap \Lambda$ and
$h^u=\text{\rm dim}_HW^u_{\varepsilon}(x)\cap \Lambda$. We will call
$h^s$ and $h^u$ the \textit{local stable} and \textit{unstable Hausdorff
dimensions} of $\Lambda$, respectively.
\\[2mm]
\textbf{V. Global Hausdorff Dimension.} Moreover, the Hausdorff
dimension of $\Lambda$ is equal to its box counting dimension and
$$
\dim_H \Lambda = \dim_B \Lambda=h^s+h^u;
$$
see \cite{MM,PV}.
\\[2mm]
\textbf{VI. Continuity of the Hausdorff Dimension.} The local
Hausdorff dimensions $h^s(\Lambda)$ and $h^u(\Lambda)$ depend
continuously on $f:M\to M$ in the $C^1$-topology; see \cite{MM,PV}.
Therefore, $\dim_H \Lambda_f = \dim_B \Lambda_f = h^s(\Lambda_f) +
h^u(\Lambda_f)$ also depends continuously on $f$ in the
$C^1$-topology. Moreover, for $r \ge 2$ and $C^r$-diffeomorphisms $f
: M \to M$, the Hausdorff dimension of a hyperbolic set $\Lambda_f$
is a $C^{r-1}$ function of $f$; see \cite{M}.
\\[3mm]
\textit{Remark.} For hyperbolic sets in dimension greater than two,
most of these properties do not hold in general; see \cite{P} for
more details.

\subsection{Implications for the Trace Map and the Spectrum}

Due to Theorem~\ref{Casdagli}, the properties I--VI can all be applied to the
hyperbolic set $\Omega_{\lambda}$
of the trace map $f_{\lambda} : \mathbb{S}_{\lambda^2} \to \mathbb{S}_{\lambda^2}$
for every $\lambda \ge 16$.
One can extract the following statement from the material in
\cite[Section~2]{C}.

\begin{lemma}
For $\lambda\ge 16$ and every $x\in \Omega_{\lambda}$, the stable
manifold $W^s(x)$ intersects the line $l_{\lambda}$ transversally.
\end{lemma}

The existence of a $C^1$-foliation $\mathcal{F}^s$ allows us to
locally consider the set $W^s(\Omega_{\lambda})\cap l_{\lambda}$ as
a $C^1$-image of a set $W^u_{\varepsilon}(x)\cap\Omega_{\lambda}$.
Therefore, we obtain the following consequences for the spectrum of
the Fibonacci Hamiltonian.

\begin{theorem}
For $\lambda\ge 16$, the following statements hold:
\\
{\rm (a)} The spectrum $\Sigma_\lambda$ depends continuously on
$\lambda$ in the Hausdorff metric.
\\
{\rm (b)} We have $\dim_H (\Sigma_\lambda) = \dim_B
(\Sigma_\lambda)$.
\\
{\rm (c)} For every small $\varepsilon>0$ and every $E \in
\Sigma_\lambda$, we have
$$
\dim_H \left( (E-\varepsilon, E+\varepsilon) \cap \Sigma_\lambda
\right) = \dim_H (\Sigma_\lambda)
$$
and
$$
\dim_B \left( (E-\varepsilon, E+\varepsilon) \cap \Sigma_\lambda
\right) = \dim_B (\Sigma_\lambda).
$$
{\rm (d)} The Hausdorff dimension $\dim_H (\Sigma_\lambda)$ is a
$C^{\infty}$-function of $\lambda$.
\end{theorem}

In particular, part (b) establishes Theorem~\ref{dimeqthm} as
formulated in the Introduction.
\end{appendix}


\begin{thebibliography}{99}

\bibitem{as} M.\ Abramowitz and I.\ Stegun, \textit{Handbook of Mathematical
Functions with Formulas, Graphs,
and Mathematical Tables}, Dover, New York, 1965

\bibitem{BGT} J.-M.\ Barbaroux, F.\ Germinet, and S.\ Tcheremchantsev, Fractal
dimensions and the
phenomenon of intermittency in quantum dynamics, \textit{Duke Math.\
J.}\ \textbf{110} (2001), 161--193

\bibitem{BGT2} J.-M.\ Barbaroux, F.\ Germinet, and S.\ Tcheremchantsev,
Generalized fractal dimensions:
equivalences and basic properties, \textit{J.\ Math.\ Pures Appl.}\
\textbf{80} (2001), 977--1012

\bibitem{bist} J.\ Bellissard, B.\ Iochum, E.\ Scoppola, and D.\ Testard,
Spectral properties of one-dimensional quasicrystals,
\textit{Commun.\ Math.\ Phys.}\ {\bf 125} (1989), 527--543

\bibitem{C} M.\ Casdagli, Symbolic dynamics for the renormalization map of
a quasiperiodic Schr\"odinger equation, \textit{Commun.\ Math.\
Phys.}\ \textbf{107} (1986), 295--318

\bibitem{d1} D.\ Damanik, $\alpha$-continuity properties of one-dimensional
quasicrystals, \textit{Commun.\
Math.\ Phys.}\ \textbf{192} (1998), 169--182

\bibitem{D} D.\ Damanik, Strictly ergodic subshifts and associated operators,
in \textit{Spectral Theory and Mathematical Physics: A Festschrift
in Honor of Barry Simon's 60th Birthday}, 505--538, Proceedings of
Symposia in Pure Mathematics \textbf{74}, American Mathematical
Society, Providence, 2006

\bibitem{dkl} D.\ Damanik, R.\ Killip, and D.\ Lenz, Uniform spectral properties
of one-dimensional
quasicrystals. III. $\alpha$-continuity, \textit{Commun.\ Math.\ Phys.}\ {\bf
212}
(2000), 191--204

\bibitem{dl} D.\ Damanik and D.\ Lenz, Uniform spectral properties of
one-dimensional quasicrystals. I.~Absence
of eigenvalues, \textit{Commun.\ Math.\ Phys.}\ \textbf{207}
(1999), 687--696

\bibitem{dl2} D.\ Damanik and D.\ Lenz, Uniform spectral properties of
one-dimensional quasicrystals. II.~The Lyapunov
exponent, \textit{Lett.\ Math.\ Phys.}\ \textbf{50} (1999),
245--257

\bibitem{DT} D.\ Damanik and S.\ Tcheremchantsev, Power-Law bounds on transfer
matrices and
quantum dynamics in one dimension, \textit{Commun.\ Math.\ Phys.}
{\bf 236} (2003), 513--534

\bibitem{DT1} D.\ Damanik and S.\ Tcheremchantsev, Scaling estimates for
solutions and dynamical lower bounds on
wavepacket spreading, \textit{J.\ Anal.\ Math.}\ \textbf{97}
(2005), 103--131

\bibitem{DT2} D.\ Damanik and S.\ Tcheremchantsev, Upper bounds in quantum
dynamics, \textit{J.\ Amer.\ Math.\ Soc.}\ \textbf{20} (2007),
799--827

\bibitem{DT3} D.\ Damanik and S.\ Tcheremchantsev, in preparation

\bibitem{GKT} F.\ Germinet, A.\ Kiselev, and S.\ Tcheremchantsev, Transfer
matrices and transport for
Schr\"odinger operators, \textit{Ann.\ Inst.\ Fourier} (\textit{Grenoble}),
\textbf{54}
(2004), 787--830

\bibitem{H} B.\ Hasselblatt, Hyperbolic dynamical systems, in \textit{Handbook
of Dynamical
Systems, Vol. 1A}, 239--319, North-Holland, Amsterdam, 2002

\bibitem{hks} A.\ Hof, O.\ Knill, and B.\ Simon, Singular continuous spectrum
for palindromic Schr\"dinger
operators, \textit{Commun.\ Math.\ Phys.}\ \textbf{174} (1995),
149--159

\bibitem{it} B.\ Iochum and D.\ Testard, Power law growth for the resistance in
the Fibonacci
model, \textit{J.\ Stat.\ Phys.} \textbf{65} (1991), 715--723

\bibitem{jl} S.\ Jitomirskaya and Y.\ Last, Power law subordinacy and singular
spectra. II.~Line operators
\textit{Commun.\ Math.\ Phys.}\ \textbf{211} (2000), 643--658

\bibitem{K} L.\ Kadanoff, Analysis of cycles for a volume preserving map,
unpublished

\bibitem{ka} M.\ Kaminaga, Absence of point spectrum for a class of discrete
Schr\"odinger operators with quasiperiodic potential, \textit{Forum
Math.}\ \textbf{8} (1996), 63--69

\bibitem{kkl} R.\ Killip, A.\ Kiselev, and Y.\ Last, Dynamical upper bounds on
wavepacket
spreading, \textit{Amer.\ J.\ Math.} {\bf 125} (2003), 1165--1198

\bibitem{kkt} M.\ Kohmoto, L.\ P. Kadanoff, and C.\ Tang, Localization problem
in one dimension:
Mapping and escape, \textit{Phys.\ Rev.\ Lett.} \textbf{50} (1983), 1870--1872

\bibitem{ko} S.\ Kotani, Jacobi matrices with random potentials taking finitely
many values,
\textit{Rev.\ Math.\ Phys.}\ \textbf{1} (1989), 129--133

\bibitem{lw} Q.-H.\ Liu and Z.-Y.\ Wen, Hausdorff dimension of spectrum of
one-dimensional Schr\"odinger operator with
Sturmian potentials, \textit{Potential Anal.}\ \textbf{20} (2004),
33--59

\bibitem{M} R.\ Ma\~n\'e, The Hausdorff dimension of horseshoes of
diffeomorphisms of surfaces, \textit{Bol.\ Soc.\ Brasil.\ Mat.}\
(N.S.) \textbf{20} (1990), 1--24

\bibitem{matt} P.\ Mattila, \textit{Geometry of Sets and Measures in Euclidean Spaces.
Fractals and Rectifiability}, Cambridge University Press, Cambridge,
1995

\bibitem{MM} H.\ McCluskey and A.\ Manning, Hausdorff dimension for horseshoes,
\textit{Ergodic Theory Dynam.\ Systems}, \textbf{3} (1983),
251--261; Erratum, \textit{Ergodic Theory Dynam.\ Systems},
\textbf{5} (1985), 319

\bibitem{oprss} S.\ Ostlund, R.\ Pandit, D.\ Rand, H.\ J.\ Schellnhuber, and E.\
D.\ Siggia,
One-dimensional Schr\"odinger equation with an almost periodic potential,
\textit{Phys.\
Rev.\ Lett.} \textbf{50} (1983), 1873--1877

\bibitem{PT} J.\ Palis and F.\ Takens, \textit{Hyperbolicity and Sensitive
Chaotic
Dynamics at Homoclinic Bifurcations. Fractal Dimensions and
Infinitely Many Attractors}, Cambridge University Press, Cambridge,
1993

\bibitem{PV} J.\ Palis and M.\ Viana, On the continuity of the Hausdorff
dimension and limit capacity for horseshoes, in \textit{Dynamical
Systems}, Lecture Notes in Mathematics \textbf{1331}, 150--160,
Springer, Berlin, 1988

\bibitem{P} Ya.\ Pesin, \textit{Dimension Theory in Dynamical Systems},
University of Chicago Press, Chicago, 1997

\bibitem{r} L.\ Raymond, A constructive gap labelling for the discrete
Schr\"odinger operator on a quasiperiodic chain, preprint (1997)

\bibitem{su} A.\ S\"ut\H{o}, The spectrum of a quasiperiodic Schr\"odinger
operator, \textit{Commun.\
Math.\ Phys.}\ {\bf 111} (1987), 409--415

\bibitem{su2} A.\ S\"ut\H{o}, Singular continuous spectrum on a Cantor set of
zero Lebesgue measure
for the Fibonacci Hamiltonian, \textit{J.\ Statist.\ Phys.} \textbf{56} (1989),
525--531

\bibitem{T} F.\ Takens, Limit capacity and Hausdorff dimension of
dynamically defined Cantor sets, in \textit{Dynamical Systems},
Lecture Notes in Mathematics \textbf{1331}, 196--212, Springer,
Berlin, 1988

\bibitem{T2} S.\ Tcheremchantsev, Mixed lower bounds for quantum transport,
\textit{J.\ Funct.\
Anal.} {\bf 197} (2003), 247--282

\bibitem{T3} S.\ Tcheremchantsev, in preparation

\end{thebibliography}
\end{document}